\documentclass[12pt]{amsart}
\usepackage{amsmath, amssymb, amsthm, mathrsfs, mathtools, amscd}
\usepackage[latin1]{inputenc}
\usepackage{hyperref}
\usepackage[all]{xy}

\renewenvironment{proof}{{\noindent \bfseries Proof:}}{}
\newtheorem{theorem}{Theorem}[section]
\newtheorem{lemma}[theorem]{Lemma}

\newtheorem{proposition}[theorem]{Proposition}
\newtheorem{definition}[theorem]{Definition}
\newtheorem{remark}[theorem]{Remark}

\renewcommand{\qed}{\hfill \ensuremath{\Box}}

\newcommand\mc[1]{\mathcal{#1}}
\newcommand\BL{\mc{B}(L)}
\newcommand\cH{\mc{H}}
\newcommand\PH{\mc{P}(\cH)}
\newcommand\BH{\mc{L}(\cH)}
\newcommand\VH{\mc{V}(\cH)}
\newcommand\PV{\mc{P}(V)}
\newcommand\bbC{\mathbb{C}}
\newcommand\join{\vee}
\newcommand\cN{\mc N}

\newcommand\ra{\rightarrow}
\newcommand\bbR{\mathbb{R}}

\begin{document}

\title[Reconstructing an atomic orthomodular lattice]{Reconstructing an atomic orthomodular lattice from the poset of its Boolean sublattices}
\author{Carmen Constantin}
\address{Carmen Constantin\\
Quantum Group\\
Department of Computer Science\\
University of Oxford}
\email{i.m.carmen@gmail.com}
\author{Andreas D\"oring}
\address{Andreas D\"oring\\
Clarendon Laboratory\\
Department of Physics\\
University of Oxford}
\email{doering@atm.ox.ac.uk}
\date{5. December 2013}

\begin{abstract} We show that an atomic orthomodular lattice $L$ can be reconstructed up to isomorphism from the poset $\BL$ of Boolean subalgebras of $L$. A motivation comes from quantum theory and the so-called topos approach, where one considers the poset of Boolean sublattices of $L=\PH$, the projection lattice of the algebra $\BH$ of bounded operators on Hilbert space.
\end{abstract}

\maketitle

\vspace{0.5cm}

\section{Introduction} Orthomodular lattices  play a prominent role in quantum logic \cite{DCG02,Var07}. Often, existence of atoms is required for conceptual reasons. The prototypical example is $\PH$, the lattice of projections on a Hilbert space $\cH$, which is an atomic, complete orthomodular lattice.

In recent years, a new form of logic for quantum systems has been developed in the so-called topos approach to quantum theory \cite{DI08a,DI08b,DI08c,DI08d,DI11}. This is based on presheaves over the set of abelian von Neumann subalgebras of a von Neumann algebra $\cN$. For clarity and simplicity, consider the example $\cN=\BH$, the algebra of bounded linear operators on the Hilbert space $\cH$. An abelian von Neumann subalgebra $V\subset\BH$ has a lattice of projections $\PV$ that is a complete Boolean algebra. 

In the topos approach, one considers the set $\mc V(\BH)$ of all abelian von Neumann subalgebras of $\BH$ and partially orders this set under inclusion. Correspondingly, we have the set $\mc B(\PH)$ of complete Boolean sublattices of $\PH$, the projection lattice of $\BH$. The abelian subalgebras $V\in\mc V(\BH)$ and their corresponding complete Boolean sublattices $\PV\in\mc B(\PH)$ are called \emph{contexts}. Conceptually, they can be thought of as `classical perspectives' on the quantum system. 

The topos-based form of logic for quantum systems uses some presheaf constructions over the poset $\mc V(\BH)$ that are beyond the scope of this paper (see \cite{DI08b,Doe12}). Here, we focus on a simpler question: can the orthomodular lattice $\PH$, which is traditionally used in quantum logic, be reconstructed from the poset of contexts $\mc B(\PH)$ that underlies the constructions in the topos approach?

The answer is affirmative, and in fact, we will show that for any atomic orthomodular lattice $L$, one can reconstruct (up to isomorphism) the set of elements, the order relation and the orthocomplementation of $L$ from the poset $\mc B(L)$ of Boolean subalgebras of $L$. Conceptually, this means that by considering the partially ordered set of contexts, one does not lose information compared to considering the whole orthomodular lattice. This also implies that the new form of presheaf- and topos-based form of logic for quantum systems is (at least) as rich as traditional quantum logic.

\emph{Acknowledgements.} We thank John Harding, Rui Soares Barbosa, Nadish de Silva, Dan Marsden and Andrei Constantin for feedback.

\section{Relation to earlier results} Apart from the connections to quantum logic and the topos approach alluded to in the introduction, the main result of the paper (Thm. \ref{Thm_MainResult} below) can more directly be seen as the `object counterpart' to the result of \cite{HarNav11} (used in \cite{HarDoe10,Doe12e}) that every order-automorphism of the poset $\mc B(L)$ of an orthomodular lattice $L$ corresponds to a unique automorphism of $L$ and vice versa. Here, we consider objects (lattices), not morphisms. 

\begin{remark}
If $L=\PH$, the projection lattice on a Hilbert space $\cH$, then the height of the poset $\VH$ of abelian subalgebras equals the dimension of $\cH$ (if we include the trivial subalgebra $V_0=\bbC\hat 1$ in $\VH$; otherwise, the dimension of $\cH$ equals the height of $\VH$ plus $1$). The dimension of $\cH$ determines the Hilbert space $\cH$ up to isomorphism, and hence determines $\PH$ up to isomorphism. This is a cheap (and rather indirect) way of `reconstructing' $\PH$ from $\VH$.
\end{remark}
In this paper, we give a more explicit and more generally applicable (re)construction.

\section{Grouping and splitting} Let $L$ be an atomic orthomodular lattice with $0$ and $1$, and let $\BL$ be the set of Boolean subalgebras (BSAs) of $L$, partially ordered under inclusion. Two elements $P$ and $Q$ of $L$ are orthogonal if $P\leq Q^{\perp}$, where $Q^{\perp}$ denotes the orthocomplement of $Q$, and $Q\leq P^{\perp}$. Orthogonality also implies that the meet of $P$ and $Q$ is equal to $0$. It is clear that every element $P$ of $L$ is contained in at least one Boolean subalgebra $V$ of $L$ (for example in $V_P=\{0,P,P^{\perp},1\}$).

Let $\mc F=\{P_1,P_2,\ldots,P_n\ldots\}$ be a (possibly infinite) family of pairwise orthogonal elements in $L$ with join $1$. Then $\mc F$ generates an atomistic BSA $V\subseteq L$. The elements in $\mc F$ are the atoms of the $V$ and since each element of $V$ is a join of elements in $\mc F$, $V$ is an atomistic BSA. 

We say that a BSA $V$ generated by a family $\mc F$ as above has \textbf{dimension} $n=\#\mc F$, the cardinality of $\mc F$. In general, not every BSA $V$ of an atomic orthomodular lattice $L$ is generated by a family $\mc F$ of pairwise orthogonal elements,\footnote{An example is $\PH$, the projection lattice on an infinite-dimensional Hilbert space, which has complete Boolean sublattices that have no atoms at all, e.g. the projection lattice of the abelian von Neumann algebra generated by the position operator. There also are Boolean sublattices of $\PH$ that have some atoms, but are not generated by them.} but each element of $L$ is contained in some BSA generated by a family $\mc F$ of pairwise orthogonal elements. 

From now on, we will only consider those BSAs in $\BL$ which are generated by families of pairwise orthogonal elements. This allows us to describe inclusion relations within $\BL$ in terms of so-called grouping and splitting actions. We will write $\mc F_V$ for the family generating a BSA $V$.

\begin{definition}
If $\mc F$ and $\mc G$ are two families of pairwise orthogonal elements with join $1$, we say that $\mc G$ is obtained by \textbf{grouping} the elements in $\mc F$ if any $Q\in\mc G$ can be written as a join of elements in $\mc F$. Let $S_Q$ denote the set of elements in $\mc F$ that have join $Q$. The fact that the elements in $\mathcal{G}$ are pairwise orthogonal implies that the sets $S_Q,\;Q\in\mc G,$ are pairwise disjoint. If $\mc G$ is obtained by grouping the elements in $\mc F$, we say that $\mc F$ is obtained by \textbf{splitting} the elements in $\mc G$.
\end{definition}

The BSAs contained in a BSA $V$ are obtained from grouping the elements in $\mc F_V$ while the algebras which contain $V$, if they exist, are obtained from $V$ by splitting the elements in $\mc F_V$. 

A $2$-dimensional BSA $V\subseteq L$ is generated by two complementary elements. We can find out from the order relations within $\BL$ when one (or both) of these elements are atoms. This result will be useful later in our reconstruction of the lattice $L$. 

\begin{lemma}\label{3.2}
Given an atomistic ortholattice $L$ and a $2$-dimensional BSA $V$ of $L$, we have three possible scenarios:
\begin{itemize}
\item[i)] if $V$ is maximal in $\BL$ then its generating elements are complementary atoms. 

\item[ii)] if $V$ is included in a $3$-dimensional BSA which is maximal in $\BL$ $W$ then $V$ is generated by an atom of $L$ together with its complement which is a join of two atoms in $L$. Moreover, $W$ contains precisely two other $2$-dimensional BSAs, apart from $V$ itself.

\item[iii)]if $V$ is neither maximal, nor included in a maximal BSA, then $V$ contains an atom of $L$ if and only if all $4$-dimensional BSAs $W\subseteq L$ which contain $V$ also contain precisely three $3$-dimensional BSAs, $V_1$, $V_2$ and $V_3$, such that $V\subset V_i$, $i\in\{1,2,3\}$. 

\end{itemize}
\end{lemma}

\begin{proof}
For the first two statements, it is sufficient to observe that a BSA is maximal in $\BL$ if neither of its generating elements can be split. Since $L$ is atomistic, this implies that the generating elements of a maximal BSA must be atoms of $L$. This proves the first statement. 

For the second statement note that, $W$ being maximal, must be generated by three pairwise orthogonal atoms, call them $P$, $Q$ and $R$ which add up to the identity. The only BSAs included in $W$ are those generated either by $\{P,Q\vee R\}$ or $\{Q,P\vee R\}$ or $\{R,P\vee Q\}$, so $V$ must also be generated by one of these three families. 

For the third statement, let $\mc F_V=\{P, P^{\perp}\}$ denote the generating family of the $2$-dimensional BSA $V$. If $P$ is an atom of $L$ then any $4$-dimensional algebra $W$ which contains $V$ is obtained by splitting $P^{\perp}$ into three elements, since $P$ is an atom and cannot be split. Hence, $W$ has generating family $\mc F_W=\{P,Q_2,Q_3,Q_4\}$. There are precisely three sub-BSAs of $W$ which contain $V$. These are given by
\begin{align*}
			\mc F_{V_1} &= \{P,Q_2\vee Q_3,Q_4\},\\
			\mc F_{V_2} &= \{P,Q_2\vee Q_4,Q_3\},\\
			\mc F_{V_3} &= \{P,Q_3\vee Q_4,Q_2\}.
\end{align*}
Note that there are three other ways of grouping the elements in $W$ to obtain a $3$-dimensional BSA. The resulting $3$-dimensional BSAs $V_i$, $i=4,5,6$ do not contain $V$, since it is not possible to obtain the element $P_1$ by grouping the elements generating these other algebras.

On the other hand, consider a $2$-dimensional BSA given by $\mc F_{\tilde{V}}=\{Q, Q^{\perp}\}$ generated by two orthogonal elements which are not atoms. Since $Q$ is not an atom, it is possible to write it as a join of two orthogonal non-zero elements (in $L$), that is $Q=Q_1\join Q_2$. Similarly, it is possible to express $Q^{\perp}$ as the join of some orthogonal $Q_3$ and $Q_4$. The BSA $\mc F_{\tilde{W}}=\{Q_1,Q_2,Q_3,Q_4\}$ is a $4$-dimensional algebra which includes $\tilde{V}$, but only two of its sub-BSAs also contain $\tilde{V}$, namely
\begin{align*}
			\mc F_{\tilde{V}_1}=\{Q_1\vee Q_2,Q_3,Q_4\},\\
			\mc F_{\tilde{V}_2}=\{Q_1,Q_2,Q_3\vee Q_4\}.
\end{align*}
\qed

\end{proof}

\section{Spiked BSAs} 

Note that a family of pairwise orthogonal atoms of $L$ with join $1$ generates a \textbf{mBSA} (maximal Boolean subalgebra) of $L$.

\begin{definition}
A \textbf{sub-mBSA of $L$} is a BSA of $L$ generated by a family $\mc F$ of pairwise orthogonal elements with join $1$ with the property that only one element in $\mc F$ is the join of two atoms in $L$, while all others are atoms in $L$. 
\end{definition}

\begin{definition}
An algebra is \textbf{spiked} if it is either a mBSA, or is generated by a family $\mc F$ of pairwise orthogonal elements with join $1$ which contains precisely one non-atom of the lattice $L$ (we call this the \textbf{leading} element), while all other elements of $\mc F$ are atoms of $L$. If an algebra is spiked, we will say that its family $\mc F$ of generating elements is also spiked.
\end{definition}

\begin{definition}
Given a BSA $V$ which is not a mBSA, we say that a BSA $W$ is a \textbf{successor} of $V$ if $V\subsetneq W$ and there is no BSA $W'$ such that $V\subsetneq W'\subsetneq W$. We call a successor of a successor of a BSA $V$, if it exists, a double successor of $V$.
\end{definition}

Note that in terms of generating elements, if $W$ is a successor of $V$ then the family of elements generating $W$ is obtained from the family of elements generating $V$ by splitting precisely one element into two pairwise orthogonal elements. 

Since $L$ is atomic and orthomodular, such a splitting is possible whenever $V$ is not a mBSA (i.e. when its generating family contains at least one non-atom of $L$). This is because any non-atomic element $P$ of an atomic  lattice must be larger than some atom $Q$ and the orthomodularity condition then allows us to write $P$ as the join of $Q$ and $Q^{\perp}\wedge P$, which are easily seen to be pairwise orthogonal. Moreover, note that any element of $L$ can be written as a join of pairwise orthogonal atoms of $L$.

\begin{proposition}
Let $V\in\BL$ be a BSA generated by the (possibly infinite) family of elements $\mc F_V=\{P_1,P_2,\ldots,P_k,\ldots\}$. If we assume that $V$ is neither a mBSA, nor a sub-mBSA, then $V$ is spiked if and only if all double successors of $V$ contain precisely three successors of $V$. 
\end{proposition}

\begin{proof}
Completely analogous to the proof of the third statement of Lemma \ref{3.2}.\qed
\end{proof}

This result is important because it shows that the order structure of $\BL$ allows us to decide whether a given BSA $V\in\BL$ is spiked or not. Since every spiked BSA has a distinguished leading element, one can guess that we want to somehow link the elements of the lattice $L$ to the spiked BSAs of $\BL$ using the information encoded within the order structure of $BL$, which is what we will do in the following section.

Let $V$ be a $2$-dimensional BSA generated by $\mc F_V=\{P,P^{\perp}\}$, and let $\mc S_V$ be the set of spiked abelian BSAs which contain $V$. The generating family $\mc F$ of an element $\tilde V$ of $\mc S_V$ is obtained either by completely splitting $P$ into pairwise orthogonal atoms and splitting $P^{\perp}$ into a spiked family of elements, or by completely splitting $P^{\perp}$ into pairwise orthogonal atoms and splitting $P$ into a spiked family of elements.

The set $\mc S_V$ of spiked abelian BSAs which contain $V$ is partially ordered under inclusion.
\begin{itemize}
	\item [(a)] If $V$ is not spiked, the generating family $\mc F$ of a \emph{minimal} element in $\mc S_V$ with respect to this partial order is obtained by either taking $P$ as the leading element and splitting $P^{\perp}$ into atoms, or by taking $P^{\perp}$ as leading projection and splitting $P^{\perp}$ into atoms. Let $\mc M_V$ denote the set of minimal elements in $\mc S_V$. We call $\mc M_V$ the set of minimal spiked sup-BSAs of $V$ in $\BL$. 
	\item [(b)] If $V$ is spiked, the minimal element of $\mc S_V$ which contains $V$ will of course be $V$ itself. Hence for a spiked $2$-dimensional BSA $V$ we establish by convention the set $\mc M_V$ to be the set of all mBSAs which contain $V$, as these algebras correspond to keeping the atom fixed and completely splitting the co-atom, together with $V$ itself which corresponds to keeping the co-atom fixed.
		\item[(c)] if $V$ is spiked and submaximal, we again define $\mc M_V$ to be the set of all mBSAs which contain $V$ together with $V$ itself.   
	\item [(d)] If $V$ is spiked and maximal then it is generated by a pair of orthocomplementary atoms. These two atoms are  not comparable to any other elements in the lattice $L$ except for the top and bottom elements. The set $\mc M_V$ contains only one element, namely $V$ itself.  

\end{itemize}

\section{Reconstructing $L$ from $\BL$} Every $2$-dimensional BSA $V$ with $\mc F_V=\{P,P^{\perp}\}$ is generated by two complementary elements. Hence there is an obvious two-to-one mapping from $L$ to the $2$-dimensional elements of $\BL$, which of course are the atoms of the poset $\BL$. Therefore, in order to generate all the elements of the atomic orthomodular lattice $L$ from the poset $\BL$, we need to assign two elements (corresponding to the two elements $P,P^{\perp}$) to each $2$-dimensional BSA $V$ with $\mc F_V=\{P,P^{\perp}\}$. 

The minimal spiked sup-BSAs of a given $2$-dimensional BSA $V$ make good candidates for this assignment. On the one hand, they can be characterised using only information derived from the poset structure of $\BL$, on the other hand, a minimal spiked sup-BSA of $V$ can be identified with its leading element, which is one of the two generating elements of $V$. Yet, this would give us a many-to-one mapping in general, since there are many (e.g. in $\BH$ continuously many) minimal spiked sup-BSAs of $V$ with the same leading element, corresponding to the many possible ways of splitting its complement. Therefore, it will make sense to define two equivalence classes of algebras within $\mc M_V$ consisting of those algebras whose generating families of elements have the same leading element. 

In the non-degenerate cases (a-c) above, our task is to identify these two equivalence classes using the information encoded within the order structure of $\BL$. By partitioning the sets $\mc M_V$ into two equivalence classes, we are in effect identifying all pairs of complementary elements of the lattice $L$. Later we will see how the order relations between non-complementary elements can be replicated using the corresponding equivalence classes.

Of course, in the degenerate case (d) when $V$ is also maximal, we already know that $V$ is generated by two orthocomplementary atoms, and since these are not comparable to any other elements of $L$, the set $\mc M_V$ does not need any further analysis. 

For a spiked $2$-dimensional BSA $V$ with $\mc F_V=\{P,P^{\perp}\}$, where $P$ is an atom, it is easy to establish what the two equivalence classes should be. One of them, call it $\mathcal{R}_V$, ought to contain the mBSAs which contain $V$ -- this corresponds to keeping the atom $P$ as the `leading' element and completely splitting its complement $P^{\perp}$ into atoms (this is a slight abuse of terminology, since there is no leading element in a mBSA). The other equivalence class, call it $\mathcal{S}_V$, ought to contain only $V$ itself -- this corresponds to keeping the co-atom $P^{\perp}$ as the leading element.

Similarly, for a $2$-dimensional BSA $W$ generated by an element $P$ that is the join of two atoms, together with its complement $P^{\perp}$, we define one equivalence class to contain all the sub-mBSAs in $\mathcal{M}_W$ and the other one to contain all the $3$-dimensional BSAs in $\mathcal{M}_W$.

For non-spiked $2$-dimensional BSA whose (minimal) generating elements are joins of $3$ or more atoms, the two equivalence classes can be determined by considering the inclusion relations between elements belonging to different sets of minimal spiked sup-BSAs, as we will show now.
\begin{lemma}
If $V$ is a non-spiked $2$-dimensional BSA whose generating elements are joins of $3$ or more atoms, and if $A,B\in \mc M_V$, then $A$ and $B$ have the same leading element if and only if there exists some non-spiked $2$-dimensional $W\neq V$ and $C,D\in \mathcal{M}_W$ such that $A\subseteq C$ and $B\subseteq D$.
\end{lemma}

\begin{proof} Assume that $\mc F_W=\{Q,Q^{\perp}\}$ and $\mc F_V=\{P,P^{\perp}\}$ and that $A,B\in\mc M_V$ and $C,D\in\mc M_W$ such that $A\subseteq C$ and $B\subseteq D$. If $P_A, P_B, P_C$ and $P_D$ are the respective leading elements  of $A,B,C$ and $D$ (it makes sense to speak about leading elements, since $\mc M_V$ and $\mathcal{M}_W$ do not contain any mBSAs, as neither $V$ nor $W$ are spiked BSAs), the inclusion relations imply that $P_C\leq P_A$ and $P_D\leq P_B$.

Note at this point that the leading element of a minimal spiked sup-BSA of $V$ must be equal to either $P$ or $P^{\perp}$ (hence $P_A,P_B\in\{P,P^{\perp}\}$), while the leading element of a minimal spiked sup-BSA of $W$ must be equal to either $Q$ or $Q^{\perp}$ (hence $P_C,P_D\in\{Q,Q^{\perp}\}$). Assume towards a contradiction that $P_A\neq P_B$. Then $P_A$ and $P_B$ must be complementary elements. But if $P^{\perp}_A=P_B$, then $P_C$ and $P_D$ must also be complementary elements, otherwise the inclusion relations would imply that $P_A\geq P_C$ and $P^{\perp}_A\geq P_D=P_C$, which is imposible. This means that $P^{\perp}_C=P_D$. However, this leads to a contradiction, since 
\[
			P^{\perp}_C=P_D\leq P_B=P^{\perp}_A\Longrightarrow P_C\geq P_A,
\]
but $P_C=P_A$ is not possible since $W\neq V$.  Hence $P_A$ must be equal to $P_B$.

On the other hand, if $A$ and $B$ have the same leading element, then their generating families are of the form $\mc F_A=\{P, R_1,R_2,\ldots\}$ and $\mc F_B=\{P,S_1,S_2,\ldots\}$, and it is possible to write $P$ as the join of two orthogonal elements $Q$ and $Z$, where $Q$ is an atom, and $Z$ is not an atom. The BSAs $C$ and $D$ given by
\[
			\mc F_C:=\{Q,Z,S_1,S_2,\ldots\} \text{ and } \mc F_D:=\{Q,Z,R_1,R_2,\ldots\}
\]
are sup-BSAs of $A$ and $B$, respectively, and they belong to the set of minimal spiked sup-BSAs of the non-spiked $2$-dimensional BSA $W$ with $\mc F_W=\{Q,Q^{\perp}\}$. \qed
\end{proof}

Once the equivalence classes on the sets of minimal spiked sup-BSAs have been established, it is possible to define an order $\preceq$ on them which replicates the order within the lattice of elements.

\begin{definition}\label{Def1}
If $[X]$ and $[Y]$ are two equivalence classes corresponding to non-spiked BSAs, we say that $[X]\preceq [Y]$ if there exists $A\in[X]$ and $B\in [Y]$ such that $A\supseteq B$.
\end{definition}

If $A\supseteq B$ as above and $\mc F_A=\{P,R_1,R_2,\ldots\}$ with all the $R_i$ atoms while $\mc F_B=\{Q,S_1,S_2,\ldots\}$ with all the $S_i$ atoms, then the generating elements of $B$ are obtained by grouping the generating elements in $A$. This implies that the leading element of $B$ (which is the only non-atom) must be equal to a join of generating elements of $A$. This join must include the leading element of $A$ among its terms, as this is the only possible way of grouping the generating elements of $A$ into a spiked family. Hence there is some index set $I$ such that 
\[
			Q=P\vee\bigvee_{i\in I} S_i
\]
and hence $Q\geq P$.

Since all elements of $[X]$ have the same leading element and similarly, all elements of $[Y]$ have the same leading element, the order relation introduced in Definition \ref{Def1} is well defined. 

Moreover, given any two elements $Q$ and $P$ of an atomic orthomodular lattice which are neither atoms nor co-atoms, and which satisfy the order relation $Q\geq P$ within $L$, one has $Q=P\vee (P^{\perp}\wedge Q)$, and we know that $P^{\perp}\wedge Q$ can be expressed as a join of pairwise orthogonal atoms. Hence the equivalence classes corresponding to elements of $L$ which are neither atoms nor co-atoms, will always be related by $\preceq$ whenever their corresponding leading elements are related within the lattice $L$.

We have to use a different approach for defining the order relations which involve the equivalence classes corresponding to atoms and co-atoms of $L$. 

Recall first that for a spiked $2$-dimensional BSA $V$, the set $\mathcal{R}_V$ denotes the the mBSAs which contain $V$, and this is the equivalence class which corresponds to the atom of $V$, while $\mathcal{S}_V$ denotes the one member equivalence class (containing only $V$ itself) which corresponds to the co-atom of $V$.

\begin{definition}\label{Def2}
If $V$ and $W$ are $2$-dimensional BSAs such that $V$ is spiked and $W$ is not spiked, and $[X]\subset\mathcal{M}_W$, then $\mathcal{R}_V\prec [X]$ if there exists $A\in\mathcal{M}_W-[X]$ such that $V\subseteq A$. If this is the case, then also $\mathcal{S}_V\succ \mathcal{M}_W-[X]$. 
\end{definition}

Note that if $V$ with $\mc F_V=\{P,P^{\perp}\}$ is a spiked BSA with $P$ an atom, and $W$ with $\mc F_W=\{Q,Q^{\perp}\}$ is a non-spiked BSA. And if moreover $P<Q$, then $P^{\perp}>Q^{\perp}$ and  there is some way of decomposing $Q$ into a join over a set of atoms which includes $P$. Hence $V$ will be contained in some minimal spiked sup-BSA of $W$ which has $Q^{\perp}$ as its leading element. So $\preceq$ is again well-defined and it captures all the relations between atoms (or co-atoms) and other elements of $L$.

The only relations from $L$ we have not yet captured are those between the atoms and co-atoms themselves. We do this with the following definition.

\begin{definition}\label{Def3}
If both $V$ and $W$ are spiked, their generating elements will either be equal (if $V=W$) or incomparable. Then either $\mathcal{R}_V=\mathcal{R}_W$ and $\mathcal{S}_V=\mathcal{S}_W$, or they are incomparable. 
\end{definition}

\begin{theorem}			\label{Thm_MainResult}
Let $\mc B_2(L)$ denote the set of $2$-dimensional BSAs of an atomistic ortholattice $L$ which are not mBSAs. Let $\mc M_2(L)$ denote the set of $2$-dimensional mBSAs of $L$. The set $$\mc C(L):=\{\mc M_V/_\sim \}_{V\in \mc B_2(L)}\cup \{A_W^1, A_W^2\}_{W\in \mc M_2(L)} \cup \{0,1\}$$ together with the order $\preceq$ defined in $\ref{Def1}$, $\ref{Def2}$ and $\ref{Def3}$ above, and the additional conventions that $0$ and $1$ stand for the top and the bottom elements of $\mc C(L)$, while the $A_W^1$s and $A_V^2$s are pairs of orthocomplementary atoms which are only comparable with $0$ and $1$, is isomorphic to $L$
\end{theorem}

\vspace{0.7cm}

\begin{proof}

The lattice isomorphism can easily be constructed using the results presented so far. It sends the top and bottom elements of $L$ to the top and bottom elements of $\mc C(L)$. The orthocomplementary atoms of $L$ are identified with the elements of the pairs of the form $\{A_W^1, A_W^2\}$. And for all the other elements $P\in L$, if $\{P,R_1,R_2,\ldots\}$ is a spiked family of elements, we have the assignment
\[
			P\mapsto [\{P,R_1,R_2,\ldots,\}]\in\mathcal{M}_{\{P,P^{\perp}\}}
\]

\qed
\end{proof}

\end{document}